\documentclass{article}
\usepackage{graphicx}
\usepackage{amsmath}
\usepackage{amsthm}
\usepackage{amssymb}
\usepackage{url}
\usepackage{mathrsfs}
\usepackage{appendix}
\usepackage{latexsym}
\usepackage[noadjust]{cite}
\newcommand{\bma}{\left[\begin{matrix}}
\newcommand{\ema}{\end{matrix}\right]}
\newcommand{\be}{\begin{equation}}
\newcommand{\ee}{\end{equation}}
\newcommand{\BC}{\mathbb{C}}
\newcommand{\BR}{\mathbb{R}}

\newcommand{\cs}{\mathcal{S}}

 1

\newtheorem{lemma}{Lemma}
\newtheorem{definition}{Definition}
\newtheorem{proposition}{Proposition}

\DeclareGraphicsExtensions{.jpg,.pdf,.gif,.png,.eps}
\begin{document}
\title{Some comments on projective quadrics subordinate to pseudo--Hermitian spaces}

\author{Arkadiusz Jadczyk\footnote{E-mail address: arkadiusz.jadczyk@cict.fr}\smallskip \\ \emph{Center CAIROS}, \emph{Institut de Math\'{e}matiques de
    Toulouse}\\ \emph{Universit\'{e} Paul Sabatier, 31062 TOULOUSE CEDEX  9, France }}
\maketitle
\begin{abstract}
We study in some detail the structure of the projective quadric $Q'$ obtained by taking the quotient of the isotropic cone in a standard pseudo-hermitian space $H_{p,q}$
with respect to the positive real numbers $\BR^+$ and, further, by taking the quotient $\tilde{Q}=Q'/U(1).$ The case of signature $(1,1)$ serves as an illustration. ${\tilde Q}$ is studied as a compactification of $\BR\times H_{p-1,q-1}$
\end{abstract}
\section{Introduction}
This note is a result of a discussion with Pierre Angl\`es of the reasoning in \cite[pp. 209-212]{angles}. Pierre Angl\`es
subsequently published a corrected derivation \cite{angles2}, which gives, by a different method, the results presented below. The comments below contain the material referred to as {\em Comments on projective quadrics subordinate to pseud--Hermitian spaces\,} in the References section of \cite{angles2}. In their extremely clear paper \cite{wk} Woronowicz and Kopczy\,{}ski have explicitly shown the one--to--one correspondence between between null geodesics in the compactified Minkowski space  ${\tilde M}$ and isotropic lines in the pseudo-hermitian space $V\approx H_{2,2}.$ Below we study this correspondence for a general case of signature $(p,q).$ 

Let $\BC$ be the field of complex numbers, and let $C^*$ be the multiplicative group of complex numbers different from zero. Using the polar decomposition we can write $C^*=\BR^+\times U(1),$ where $\BR^+$ is the multiplicative group of positive real numbers and $U(1)$ is the circle group.\\
Let $ V $ be a complex vector space of finite dimension $n,$ equipped with a regular pseudo--hermitian form $(x,y)$ of signature $(p,q),$ $p+q=n,\; p,q\geq 1.$ and let $Q$ be the isotropic cone minus the origin:
 $$Q=\{x\in V:\; (x,x)=0,\; x\neq 0\}.$$ $Q$ is a real manifold of (real) dimension $2n-1$, and we denote by ${\tilde Q}$ the quotient manifold ${\tilde Q}=Q/\BC^*.$ Its elements are the equivalence classes: ${\tilde Q}=\{\{cx:\,c\in\BC^*\},\,x\in Q\}.$ ${\tilde Q}$ is a real submanifold of the complex projective space $P(V).$ We denote by $P$ the canonical projection $P:\; V\rightarrow P(V).$\\
The projection $P$ can be implemented in two steps: first taking the quotient with respect to $\BR^+$ to obtain $Q'=Q/\BR^+,$ then quotienting $Q'$ by $U(1)$ to obtain ${\tilde Q}.$ We denote the corresponding projections $P'$ and $\pi$ respectively. Thus we have $P=\pi\circ P',$ and ${\tilde Q}=Q'/U(1).$ $Q'$ and ${\tilde Q}$ are real compact manifolds of dimensions $2n-2$ and $2n-3$ respectively.\footnote{In fact, we will see that $Q'$ is diffeomorphic to the product of two odd--dimensional spheres $S^{2p-1}\times S^{2q-1}.$}\\
Let $T_x Q$ be the tangent space at $x\in Q.$ Then $T_xQ$ can be identified with a real vector subspace of $V$ as follows. If $\BR\ni t\mapsto x(t)\in Q$ is a path with $x(0)=x,$ then $(x(t),x(t))=0\; \forall t.$ Denoting by $X=\frac{dx(t)}{dt}|_{t=0}$ the tangent vector at $x,$ and using the Leibniz rule, we get $(X,x)+(x,X)=0,$ or
$$ X\in T_x Q\;\;\mbox{if and only if}\;\; \mbox{Re}((X,x))=0).$$
Let $T_x^\BC Q=x^\perp$ be the subspace of $T_xQ$ defined by the condition $(X,x)=0.$ Then $T_x^\BC Q$ is a hyperplane in $T_xQ$. In fact, while $T_xQ$ is a only a real vector space, $T_x^\BC Q$ carries the structure of a complex space.\\
Let $\eta_{jk}$ be the diagonal matrix $\eta_{jk}=\delta_{jk}$ for $j,k=1,...,p,$ $\eta_{jk}=-\delta_{jk}$ for $j,k=p+1,...,p+q.$ Let $H_{p,q}$ be the standard pseudo-hermitian space $\BC^{p+q}$ equipped with the scalar product
 $$f(u,v)=\sum_{j,k=1}^n \eta_{jk}\,u^j{\bar v}^k.$$ $H_{p,q}$ is the direct sum of two subspaces $H_{p,q}=H_{p,q}^+\oplus H_{p,q}^-$ spanned by the first $p$ (resp. last $q$) vectors of the standard basis.  Every orthonormal basis $\{e_j\}$, $(e_j,e_k)=\eta_{jk}$ in $V$ determines an isometry $\phi_e:\, H_{p,q}\rightarrow V$ and determines an orthogonal direct sum decomposition $V=\phi_e(H_{p,q}^+)\oplus \phi_e(H_{p,q}^-)$ into a positive and a negative subspace. We call such a decomposition a ``split''. We denote by $\cs$ the set of all splits of $V.$ Then $\cs$ is a homogeneous space (in fact, it is a K\"ahler manifold) for the unitary group $U(V),$ isomorphic to $U(p,q)/(U(p)\times U(q)).$\\

\subsection{The topology of $Q'=P'(Q).$}
Let $\{e_j\}$ be an orthonormal basis in $V,$ so that we can identify $V$ with $H^{p,q}.$ The equation of $Q$ becomes then
$$ \sum_{j=1}^p |z^j|^2=\sum_{j=p+1}^{p+q}|z^j|^2\neq 0.$$
Consider the submanifold $S_1$ of $H^{p,q}$ defined by the formula
\be \sum_{j=1}^p |z^j|^2=\sum_{j=p+1}^{p+q}|z^j|^2=1.\label{eq:spsq}\ee
Then $S_1$ is the product of two spheres $S_1=S^{2p-1}\times S^{2q-1}$ and the projection $P'$ restricted to $S_1$ is a diffeomorphism from $S_1$ to $Q'=P'(Q).$ The representation of $Q'$ as a product of two spheres will, in general, depend on the choice of the orthonormal basis, more specifically: on the split determined by the basis.\\
Somewhat more generally, let $s\in\cs$ be a split, so that $V$ is decomposed into a direct (orthogonal) sum $V=V_+\oplus V_-$ of positive and negative subspaces. Defining $\Vert x\Vert = (x,x)$ on $V_+,$ and $\Vert x\Vert=-(x,x)$ on $V_-,$ each isotropic vector $x\in Q$ decomposes into a sum $x=x_+ +x_-,$ with $x_+\in V_+,\, x_-\in V_-,$ and $\Vert x_+\Vert=\Vert x_-\Vert=R(x)>0.$ Rescaling $x\mapsto x/R(x)$ we get the unique representative of the equivalence class $\BR^+ x$ of $x$ with $R(x)=1.$ In other words, if we define
 $$Q_s=\{x\in Q: \Vert x_+\Vert=\Vert x_-\Vert=1\},$$ then $Q_s$ defines a global cross section of the projection $P':Q\rightarrow Q',$ and a diffeomorphism of $Q'$ onto the product of the two unit spheres, one in $V_+$ and one in $V_-.$

\subsection{The conformal structure of $Q'$}
Let $a$ be a point of $Q'$ and let $x$ be an isotropic vector in $Q$ with $P(x)=a.$ The tangent space $T_x Q$ is equipped with the (real) bilinear form $f_x(X,Y)=\mbox{Re}((X,Y)).$  Notice that the vector $x$ itself can be considered as an element of $T_x Q,$ and that the line $\BR x$ is the radical of the bilinear form $f_x,$ and is the kernel of the tangent map $(dP')_x:$ $$ \BR=\{y\in T_xQ:f_x(y,z)=0\;\forall z\in T_xQ\}=\{y\in T_xQ: (dP')_x(y)=0\}.$$ It follows that the form $f_x$ induces a regular bilinear form, which we denote  $g_x$ on the tangent space $T_aQ'.$
\begin{lemma}
With the notation as above, if $\lambda>0$ then $g_{\lambda x}=\lambda^2 g_x.$
\end{lemma}
\begin{proof}Let ${a_1}(t),{a_2}(t),$ ${a_1}(0)={a_2}(0)=a,$ be two paths in $Q'$ with tangent vectors ${\dot a}_1(0)$ and ${\dot a}_2(0)$ respectively. Let $x_1,x_2$ be the lifts: $P'(x_1(t))=a_1(t),$ $P'(x_2(t))=a_2(t),$ $x_1(0)=x_2(0)=x.$ Then, according to the definition of $g_x,$ we have that\footnote{Notice that this expression does not depend on the choice of the lifts, the reason being that vectors tangent to two lifts will differ by vectors in the kernel of $(dP')_x,$ that is in $\BR x,$ which is orthogonal to all vectors in $T_xQ.$ } $$ g_x({\dot a}_1(0),{\dot a}_2(0))=f_x({\dot x}_1(0),{\dot x}_2(0)).$$
Let $\lambda>0,$ then ${x'}_1(t)=\lambda{x}_1(t)$ and ${x'}_2(t)=\lambda{x}_2(t)$ are lifts  through $\lambda x$ of $a_1(t)$ and $a2(t)$ respectively, with tangent vectors $\lambda{\dot{x}}_1,\,\lambda{\dot{x}}_2.$ Thus\\ $g_{\lambda x}({\dot a}_1(0),{\dot a}_2(0))=\lambda^2 g_x({\dot a}_1(0),{\dot a}_2(0)).$
\end{proof}
It follows from the above lemma that what is independent of the choice of $x$ in $P'^{-1}({a})$ is the conformal class of $g_x$ and, in particular, the signature, which, by construction, is $(2p-1,2q-1).$ \footnote{The fact that the conformal class of the induced metric on $Q_s$ does not depend on the choice of the split $s$ is by no means evident if the two induced metrics are computed using two different orthonormal bases related by a general $U(p,q)$ transformation and then compared.}\\

In order to proceed further on notice that we have the following, easy to prove, lemma:
\begin{lemma}
Given a vector $x\in Q$ there exists an orthonormal basis $e$ such that $x=e_1+e_n.$
\end{lemma}
\begin{proof}Take any orthonormal basis $\{e'_i\}.$ The vector $x'=e'_1+e'_n$ is isotropic. We know that the automorphism group of $V$ acts transitively on isotropic lines (see e.g., \cite[p. 74, Corollaire 2]{bourbaki1}). Let $U$ be any automorphism of $V$ with the property $x=Ux',$ and let $e_i=Ue'_i.$ Then $\{e_i\}$ is an orthonormal basis of $V$ and $x=e_1+e_n.$
\end{proof}
Each tangent space $T_x(Q)$ is also equipped with a skew-symmetric bilinear form $F_x$ defined by
$$ F_x(X,Y)=\mbox{Im}(X,Y),\; X,Y\in T_x(Q).$$
However the form $F_x$ does not descend to the quotient $Q'=Q/\BR x$ because, owing to the fact that, for instance, $ie_1$ is in $T_x(Q)$ but $F_s(x,ie_1)=1,$ we find that $x$ is not in the radical of $F_x.$ But $F_x,$ {\bf when restricted to $x^\perp$}, does descend to a skew--symmetric bilinear form on $(dP')_x(x^\perp).$\\
Because of the lemma above it is instructive to consider first the case of $p=1,\,q=1.$
\subsection{The case of the signature $(1,1).$}
With $x\in Q,$ let $(e_1,e_2)$ be the orthonormal basis $(e_1,e_1)=1,\, (e_2,e_2)=-1,$ with $x=e_1+e_2,$ and let $\{f_1,f_2,f_3,f_4\}$  $ f_1=e_1+e_2=x,$ $f_2=i(e_1+e_2),$ $f_3=i(e_1-e_2),$ $f_4=(e_1-e_2).$
 The tangent space $T_xQ$ is spanned by the vectors $\{f_1,f_2,f_3\},$ the complex orthogonal space $x^\perp$ is spanned by the vectors $\{f_1,f_2\}.$ Let $a=P'(x).$ Notice that $(dP')_x(f_1)=0,$ while $(dP')_x$ is a bijection from the plane spanned by $\{f_2,f_3\}$ onto $T_aQ'.$ Denoting $\epsilon_1=(dP')_x(ie_1),\,\epsilon_2=(dP')_x(ie_2),$ the vectors $\epsilon_1,\epsilon_2$ form an orthonormal basis in $T_a Q'$ for the induced bilinear form $g_x:$ $$ g_x(\epsilon_1,\epsilon_1)=-g_x(\epsilon_2,\epsilon_2)=1,\; g_x(\epsilon_1,\epsilon_2)=0.$$
 $Q'$ is now the torus $S^1\times S^1$ given by the formula (\ref{eq:spsq}), now becoming:
 $$ |z^1|^2=|z^2|^2=1.$$
 Writing $z^1=\cos(\phi_1)+i\sin(\phi_1),\, z^2=\cos(\phi_2)+i\sin(\phi_2),$ the pseudo Rieman\-nian metric $g_x$ of $Q',$ when expressed in the natural torus coordinates $\phi_1,\phi_2$ is diagonal $g_x=diag(1,-1).$ The action $z\mapsto \exp(i\phi)z$ of $U(1)$ on $Q$ translates to the action $(\phi_1,\phi_2)\mapsto (\phi_1+\phi,\phi_2+\phi)$ on the torus. The tangent vector to the orbit of this action at $x$ is $f_2$ that projects onto $\epsilon_1+\epsilon_2$ at $T_a Q'.$ Taking the quotient of $Q'$ by this action we get ${\tilde Q}$ as the circle ${\tilde Q}=Q/\BC x = S^1.$ The image of $x^\perp=\BC x$ by $(dP)_x$ consists of one point - the zero vector. Since $\epsilon_1+\epsilon_2$ is a null vector for the metric $g_x,$ there is no distinguished subspace transversal to the fiber, therefore no metric whatsoever is generated by $dP$ on ${\tilde Q}.$
 \subsection{The structure of ${\tilde Q}$}
 Given a split $s\in\cs,$ let $V=V_+\oplus V_-$ be the corresponding decomposition $V.$ Every vector $x\in Q$ can be then uniquely represented as $x=x_+ + x_-,$ so that $(x_+,x_+)=(x_-,x_-).$ With
$$ Q_s=\{x\in Q:\;(x_+,x_+)=(x_-,x_-)=1\},$$
the map $P:Q\rightarrow Q',$ when restricted to $Q_s,$ becomes a diffeomorphism. The $U(1)$ action $x\mapsto c x,\; |c|=1$ leaves $Q_s$ invariant.
 $(Q',\pi)$ is a $U(1)$ principal fibre bundle over ${\tilde Q}.$ Given a split $s,$ $Q'$ is endowed with the pseudo--Riemannian metric $g_s$ that is automatically $U(1)$--invariant.\\
 Given a non--degenerate pseudo-Riemannian metric on a principal bundle, the standard method of obtaining the metric on the base space is by taking the orthogonal complement to the fibers. This method works when the orthogonal complement is transversal to the fibers. Yet in our case the vectors tangent to the fibers are isotropic, therefore the orthogonal complement is not transversal to the fibers. Nevertheless, we can obtain a natural, though degenerate, scalar product $g$ on the cotangent bundle of ${\tilde Q}.$ as follows:\\

 Let $\cs$ be as split, let $a\in{\tilde Q},$ and let $\omega,\omega'$ be two one-forms in the cotangent space $T_a^*{\tilde Q}.$ Let
 $b\in Q'$ be a point in the fibre $\pi^{-1}(a).$ The pullbacks $\pi_*\omega,\pi*\omega'$ are invariant one-forms defined at the points of the fibre $\pi^{-1}(a).$ We can therefore calculate the scalar product $g_x^*(\pi_*\omega,\pi_*\omega'),$ at any point of the fibre, and, owing to the fact that the forms and the metric are invariant, the result is independent of the chosen point. Since the scalar products
 corresponding to different choices of $x$ differ only by a scale factor, the same is true about the induced contravariant symmetric
 scalar product on the cotangent bundle of ${\tilde Q}.$ The scalar product so obtained is degenerate. Indeed, any form that vanishes
 on the image $(dP')_x(x^\perp)$ is in the radical of $g_x^*.$

 There is another way of looking at this construction.\\
 Let $W_1$ be a subspace of a real vector space $W,$ and let $f_1$ be a non--degenerate symmetric bilinear form on $W_1.$ Let $\iota: W_1\rightarrow W$ be the canonical inclusion map, and let $\iota^*:W^*\rightarrow W_1^*$ be its dual. The bilinear form $f_1$ can be considered as a map $f_1:W_1\rightarrow W_1^*,$ and, since we assume it to be non--degenerate, there exists the inverse $f_1^*:W_1^*\rightarrow W_1.$ We can then define $f^*:W^*\rightarrow W$ by $$ f^*=\iota\circ f_1^*\circ\iota^*.$$ The map $f^*$ can now be considered as a bilinear form on $W^*$ and it is easy to see that, by construction, it is symmetric. Moreover, its radical consists
 of the forms $\omega\in W^*$ that vanish on the image $\iota(W_1).$\\
 \subsection{${\tilde Q}$ as a compactification of $\BR\times H_{p-1,q-1}$}
 Let $a,b\in{\tilde Q}.$ \begin{definition}We write $a\perp b$ if and only if $a=P(x),\,b=P(y),$ where $x,y\in Q,$ and $(x,y)=0.$ Given $a\in {\tilde Q}$ we define
 $$ a^\perp=\{b\in{\tilde Q}:\;a\perp b.\}.$$\end{definition}
 \noindent It can be seen that $a^\perp$ is a closed subset of ${\tilde Q}.$\\
 Let us fix $a\in {\tilde Q},$ and let $x\in Q$ be such that $P(x)=a.$ We recall that $x^\perp$ is a complex vector subspace of $V$ that carries
 a degenerate sesquilinear form inherited from the scalar product of $V,$ with radical $\BC x.$ Therefore the quotient space $$ M\buildrel\rm df\over= x^\perp/\BC x$$ carries the pseudo-Hermitian form of signature $(p-1,q-1).$ We can realize $M$ as follows: choose $u\in Q$ such that $(u,x)=1.$ \footnote{Such a choice is always possible, for instance, by using Lemma 2, we can set $x=e_1+e_n,$ $u=(e_1-e_n)/2.$} Let $M_u$ be the orthogonal complement of $\{x,u\}$ in $V.$ Then the scalar product of $V$ restricted to $M_u$ is of signature $(p-1,q-1),$ we evidently have $M_u\subset x^\perp,$ and the projection $x^\perp \rightarrow x^\perp/\BC x$ restricted to $M_u$ is easily seen to be a bijection.\\
 We will construct now a bijection $\kappa$ from $\BR\times M_u$ onto ${\tilde Q}\setminus a^\perp\subset {\tilde Q}.$\\
 Given $r\in\BR,$ $y\in M_u$ define
  $$\kappa_0(r,y)=y+u+\left(-\frac{1}{2}(y,y)+ri\right)x.$$ Notice that the coefficient in front of $x$ has the imaginary part $r.$ It is easy to see that, automatically, $\kappa_0(r,y)\in Q$ and also $(x,\kappa_0(r,y))=1.$\\
 We define now $\kappa=P\circ\kappa_0.$ It is easy to check that $\kappa$ is injective. It remains to show that it is a surjection from $\BR\times M_u$ onto ${\tilde Q}\setminus a^\perp.$ Given $b\in {\tilde Q}\setminus a^\perp,$ let $z'$ be any point in $P^{-1}(b).$
 Then, since $b$ is not in $a^\perp,$ we have that  $(z',x)=a\neq 0.$ Taking $z=z'/a,$ we still have $P(z)=b,$ but now $(z,x)=1.$ Now, $z$ can be uniquely written in the form $z=y+\alpha u +\beta x,$ where $y\in M_u,$ $\alpha,\beta\in \BC.$ From $(z,x)=1$ we find that $\alpha=1,$ and from $(z,z)=0$ we get that $\mbox{Re}(\beta)=-\frac{1}{2}(y,y).$ Putting $r=\mbox{Im}(\beta)$ we get $b=\kappa(r,y).$
 \subsubsection{The structure of $a^\perp.$}
 \begin{proposition}
 With the notation as above, $$ a^\perp\cong\{1\}\cup \BR\times (S^{p-2}\times S^{q-2})/S^1.$$
 \end{proposition}
\begin{proof}Let $\{e_j\}$ be an orthonormal basis such that $x=e_1+e_n.$ Then any vector $y\in x^\perp$ is of the form
 $$ y=\alpha x +\sum_{i=2}^{n-1}\alpha^j e_j.$$
 Such a vector $y$ is in $Q$ if and only if $\sum_{j=2}^{p}|\alpha^j|^2= \sum_{j=p+1}^{n-1}|\alpha_j|^2.$ If all $\alpha^j$ are zero, then, necessarily, $\alpha\neq 0,$ and we can choose a unique representative of the equivalence class with $\alpha=1.$ This give s the point $\{1\}.$ If at leat one of the $\alpha^j$ is non--zero, then $\sum_{j=2}^{p}|\alpha^j|^2= \sum_{j=p+1}^{n-1}|\alpha_j|^2\neq 0$ and we can use the freedom of real scaling to get $\sum_{j=2}^{p}|\alpha^j|^2= \sum_{j=p+1}^{n-1}|\alpha_j|^2=1.$ The remaining freedom of $U(1)$ gives us $(S^{p-2}\times S^{q-2})/S^1.$ The $\alpha$ coefficient remains still arbitrary, thus the result follows.\end{proof}
 \section{Acknowldegements} The author wants to express all his thanks to Pierre Angl\'es for a fruitful and constructive discussion.

\end{document}